\documentclass[a4paper,11pt]{article}
\usepackage{authblk}
\title{Approximation schemes for bounded distance problems on fractionally treewidth-fragile graphs}

\author[1]{Zden\v{e}k Dvo\v{r}\'ak\thanks{Supported by the ERC-CZ project LL2005 (Algorithms and complexity within and beyond bounded expansion) of the Ministry of Education of Czech Republic. email:~\texttt{rakdver@iuuk.mff.cuni.cz}}}
\author[2]{Abhiruk Lahiri\thanks{Supported by ISF grant 822/18 and Ariel University Post-doctoral fellowship. email:~\texttt{abhiruk@ariel.ac.il}}}
\affil[1]{Charles University, Prague, Czech Republic}
\affil[2]{Ariel University}

\usepackage{mathrsfs,xcolor,mathtools}
\usepackage{amsmath,amsthm,amsfonts,amssymb,enumerate,hyperref,algorithmic}

\newtheorem{theorem}{Theorem}
\newtheorem{corollary}[theorem]{Corollary}
\newtheorem{lemma}[theorem]{Lemma}

\DeclareMathOperator{\msol}{\mathsf{MSOL}}

\DeclareMathOperator{\ptas}{\mathsf{PTAS}}
\DeclareMathOperator{\np}{\mathsf{NP}}

\newcommand{\GG}{\mathcal{G}}

\begin{document}

\maketitle

\begin{abstract}
We give polynomial-time approximation schemes for monotone maximization problems expressible in terms of distances (up to a fixed upper bound) and efficiently solvable in graphs of bounded treewidth.  
These schemes apply in all fractionally treewidth-fragile graph classes, a property which is true for many natural graph classes with sublinear separators.  
We also provide quasipolynomial-time approximation schemes for these problems in all classes with sublinear separators.
\end{abstract}

\section{Introduction}
In this paper, we consider optimization problems such as:
\begin{itemize}
\item \textsc{Maximum} $r$-\textsc{Independent Set}, $r \in \mathbb{Z}^+$: Given a graph $G$, the objective is to find a largest subset $X \subseteq V(G)$ such that distance in $G$ between any two vertices in $X$ is at least $r$.
\item \textsc{Maximum weight induced forest}: Given a graph $G$ and an assignment $w:V(G)\to\mathbb{Z}_0^+$ of non-negative weights to vertices, the objective is to find a subset $X \subseteq V(G)$ such that $G[X]$ does not contain a cycle and subject to that, $w(X)\coloneqq\sum_{v\in X} w(v)$ is maximized.
\item \textsc{Maximum} $(F,r)$-\textsc{Matching}, for a fixed connected graph $F$ and $r \in \mathbb{Z}^+$:
Given a graph $G$, the objective is to find a largest subset $X \subseteq V(G)$ such that $G[X]$ can be partitioned into vertex-disjoint copies of $F$ such that distance in $G$ between any two vertices belonging to different copies is at least $r$.  
\end{itemize}
To be precise, to fall into the scope of our work, the problem must satisfy the following conditions:
\begin{itemize}
\item It must be a \textbf{maximization problem on certain subsets of vertices} of an input graph, possibly with non-negative weights.
That is, the problem specifies which subsets of vertices of the input graph are \emph{admissible}, and the goal is to
find an admissible subset of largest size or weight.
\item The problem must be \textbf{defined in terms of distances between the vertices, up to some fixed bound}.  That is, there
exists a parameter $r\in \mathbb{Z}^+$ such that for any graphs $G$ and $G'$, sets $X\subseteq V(G)$ and $X'\subseteq V(G')$,
and a bijection $f:X\to X'$, if $\min(r,d_G(u,v))=\min(r,d_{G'}(f(u),f(v)))$ holds for all $u,v\in X$, then $X$ is
admissible in $G$ if and only if $X'$ is admissible in $G'$.
\item The problem must be \textbf{monotone} (i.e., all subsets of an admissible set must be admissible), or at least \textbf{near-monotone} (as happens for example for \textsc{Maximum} $(F,r)$-\textsc{Matching}) in the following sense: There exists a parameter $c\in \mathbb{Z}^+$ such that for any admissible set $A$ in a graph $G$, there exists a system $\{R_v\subseteq A:v\in A\}$ of subsets of $A$ such that every vertex belongs to $R_v$ for at most $c$ vertices $v\in A$, $v\in R_v$ for each $v\in A$, and for any $Z\subseteq A$, the subset $X\setminus \bigcup_{v\in Z} R_v$ is admissible in $G$.
\item The problem must be \textbf{tractable in graphs of bounded treewidth}, that is, there must exist a function $g$ and a polynomial $p$ such that given any graph $G$, its tree decomposition of width $t$, an assignment $w$ of non-negative weights to the vertices of $G$, and a set $X_0\subseteq X$, it is possible to find a maximum-weight admissible subset of $X_0$ in time $g(t)p(|V(G)|)$.
\end{itemize}
Let us call such problems \emph{$(\le\!r)$-distance determined $c$-near-monotone $(g,p)$-tw-tractable}.
Note that a convenient way to verify these assumptions is to show that the problem is expressible in \emph{solution-restricted Monadic Second-Order Logic} ($\msol$) \emph{with bounded-distance predicates}, i.e., by a $\msol$ formula with one free variable $X$ such that the quantification is restricted to subsets and elements of $X$, and using binary predicates $d_1$, \ldots, $d_r$, where $d_i(u,v)$ is interpreted as testing whether the distance between $u$ and $v$ in the whole graph is at most $i$. 
This ensures that the problem is $(\le\!r)$-distance determined, and $(g,O(n))$-tw-tractable for some function $g$ by Courcelle's metaalgorithmic result~\cite{Courcelle90}.

Of course, the problems satisfying the assumptions outlined above are typically hard to solve optimally, even in rather restrictive circumstances. 
For example, \textsc{Maximum Independent Set} is $\np$-hard even in planar graphs of maximum degree at most $3$ and arbitrarily large (fixed) girth~\cite{AlekseevLMM08}. 
Moreover, it is hard to approximate it within factor of $0.995$ in graphs of maximum degree at most three~\cite{BermanK99}. Hence, to obtain polynomial-time approximation schemes ($\ptas$), i.e., polynomial-time algorithms for approximating within any fixed precision, further restrictions on the considered graphs are needed.

A natural restriction that has been considered in this context is the requirement that the graphs have sublinear separators (a set $S$ of vertices of a graph $G$ is a \emph{balanced separator} if every component of $G\setminus S$ has at most $|V(G)|/2$ vertices, and a hereditary class $\GG$ of graphs has \emph{sublinear separators} if for some $c<1$, every graph $G \in \GG$ has a balanced separator of size $O(|V(G)|^c)$). This restriction still lets us speak about many interesting graph classes (planar graphs~\cite{LiptonT79} and more generally proper minor-closed classes~\cite{AlonST90}, many geometric graph classes~\cite{MillerTTV97}, \ldots).
Moreover, the problems discussed above admit $\ptas$ in all classes with sublinear separators or at least in substantial subclasses of these graphs:
\begin{itemize}
\item \textsc{Maximum Independent Set} has been shown to admit $\ptas$ in graphs with sublinear separators already in the foundational paper of Lipton and Tarjan~\cite{LiptonT80}.
\item For any positive integer, \textsc{Maximum} $r$-\textsc{Independent Set} and several other problems are known to admit $\ptas$ in graphs with sublinear separators by a straightforward local search algorithm~\cite{Har-PeledQ17}.
\item All of the problems mentioned above (an more) are known to admit $\ptas$ in planar graphs by a layering argument of Baker~\cite{Baker94}; this approach can be extended to some related graph classes, including all proper minor-closed classes~\cite{DawarGKS06,Dvorak20}.
\item The problems also admit $\ptas$ in graph classes that admit thin systems of overlays~\cite{Dvorak18}, a technical property satisfied by all proper minor-closed classes and by all hereditary classes with sublinear separators and bounded maximum degree.
\item Bidimensionality arguments~\cite{DemaineH05} apply to a wide range of problems in proper minor-closed graph classes.
\end{itemize}
However, each of the outlined approaches has drawbacks. On one side, the local search approach only applies to specific problems and does not work at all in the weighted setting. 
On the other side of the spectrum, Baker's approach is quite general as far as the problems go, but there are many hereditary graph classes with sublinear separators to which it does not seem to apply. 
The approach through thin systems of overlays tries to balance these concerns, but it is rather technical and establishing this property is difficult.

Another option that has been explored is via \emph{fractional treewidth-fragility}.  For a function $f \colon \mathbb{Z}^+ \times \mathbb{Z}^+ \to \mathbb{Z}^+$ and a polynomial $p$, a class of graphs $\GG$ is \emph{$p$-efficiently fractionally treewidth-$f$-fragile} if there exists an algorithm that for every $k \in \mathbb{Z}^+$ and a graph $G \in \GG$ returns in time $p(|V(G)|)$ a collection of subsets $X_1, X_2, \dots X_m \subseteq V(G)$ such that each vertex of $G$ belongs to at most $m/k$ of the subsets, and moreover, for $i=1,\ldots,m$, the algorithm also returns a tree decomposition of $G \setminus X_i$ of width at most $f(k, |V(G)|)$.
We say a class is \emph{$p$-efficiently fractionally treewidth-fragile} if $f$ does not depend on its second argument (the number of vertices of $G$).  This property turns out to hold for basically all known natural graph classes with sublinear separators.
In particular, a hereditary class $\GG$ of graphs is efficiently fractionally treewidth-fragile if
\begin{itemize}
\item $\GG$ has sublinear separator and bounded maximum degree~\cite{Dvorak16},
\item $\GG$ is proper minor-closed~\cite{DeVosDOSRSV04, Dvorak20}, or
\item $\GG$ consists of intersection graphs of convex objects with bounded aspect ratio in a finite-dimensional Euclidean space and the graphs have bounded clique number, as can be seen by a modification of the argument of Erlebach et al.~\cite{ErlebachJS05}. This includes all graph classes with polynomial growth~\cite{KrauthgamerL07}.
\end{itemize}
In fact, Dvo\v{r}\'ak conjectured that every hereditary class with sublinear separators is fractionally treewidth-fragile, and gave the following result towards this conjecture.
\begin{theorem}[Dvo\v{r}\'ak~\cite{Dvorak18a}]\label{thm-pll}
There exists a polynomial $p$ so that the following claim holds. For every hereditary class $\GG$ of graphs with sublinear separators, there exists a polynomial $q$ such that $\GG$ is $p$-efficiently fractionally treewidth-$f$-fragile for the function $f(k,n)=q(k\log n)$.
\end{theorem}
Moreover, Dvo\v{r}\'ak~\cite{Dvorak16} observed that weighted \textsc{Maximum Independent Set} admits a $\ptas$ in any efficiently fractionally treewidth-fragile class of graphs. 
Indeed, the algorithm is quite simple, based on the observation that for the sets $X_1$, \ldots, $X_m$ from the definition of fractional treewidth-fragility, at least one of the graphs $G \setminus X_1$, \ldots, $G \setminus X_m$ (of bounded treewidth) contains an independent set whose weight is within the factor of $1-1/k$ from the optimal solution. 
A problem with this approach is that it does not generalize to more general problems; even for the \textsc{Maximum $2$-Independent Set} problem, the approach fails, since a $2$-independent set in $G \setminus X_i$ is not necessarily $2$-independent in $G$. 
Indeed, this observation served as one of the motivations behind more restrictive (and more technical) concepts employed in~\cite{Dvorak18,Dvorak20}.

As our main result, we show that this intuition is in fact false: There is a simple way how to extend the approach outlined in the previous paragraph to all bounded distance determined near-monotone tw-tractable problems.
\begin{theorem}\label{thm-main}
For every class $\GG$ of graphs with bounded expansion, there exists a function $h:\mathbb{Z}^+\times\mathbb{Z}^+\to\mathbb{Z}^+$ such that the following claim holds. 
Let $c$ and $r$ be positive integers, $g:\mathbb{Z}^+\to\mathbb{Z}^+$ and $f:\mathbb{Z}^+\times\mathbb{Z}^+\to\mathbb{Z}^+$ functions and $p$ and $q$ polynomials. 
If $\GG$ is $q$-efficiently fractionally treewidth-$f$-fragile, then for every $(\le\!r)$-distance determined $c$-near-monotone $(g,p)$-tw-tractable problem, there exists an algorithm that given a graph $G\in \GG$, an assignment of non-negative weights to vertices, and a positive integer $k$, returns in time $h(r,c)|V(G)|+q(|V(G)|)\cdot p(|V(G)|)\cdot g(f(h(r,c)k,|V(G)|))$ an admissible subset of $V(G)$ whose weight is within the factor of $1-1/k$ from the optimal one.
\end{theorem}

Note that the assumption that $\GG$ has bounded expansion is of little consequence---it is true for any hereditary class with sublinear separators~\cite{DvorakN16} as well as for any fractionally treewidth-fragile class~\cite{Dvorak16}; see Section~\ref{sec-dist} for more details. 
The time complexity of the algorithm from Theorem~\ref{thm-main} is polynomial if $f$ does not depend on its second argument, and quasipolynomial (exponential in a polylogaritmic function) if $f$ is logarithmic in the second argument and $g$ is single-exponential (i.e., if $\log \log g(n)=O(\log n)$). 
Hence, we obtain the following corollaries.

\begin{corollary}
Let $c$ and $r$ be positive integers, $g:\mathbb{Z}^+\to\mathbb{Z}^+$ a function and $p$ a polynomial. 
Every $(\le\!r)$-distance determined $c$-near-monotone $(g,p)$-tw-tractable problem admits a $\ptas$ in any efficiently fractionally treewidth-fragile class of graphs.
\end{corollary}

We say a problem admits a quasipolynomial-time approximation schemes ($\mathsf{QPTAS}$) if there exist quasipolynomial-time algorithms for approximating the problem within any fixed precision.  
Combining Theorems~\ref{thm-pll} and \ref{thm-main}, we obtain the following result.

\begin{corollary}
Let $c$ and $r$ be positive integers, $g:\mathbb{Z}^+\to\mathbb{Z}^+$ a single-exponential function, and $p$ a polynomial.  
Every $(\le\!r)$-distance determined $c$-near-monotone $(g,p)$-tw-tractable problem admits a $\mathsf{QPTAS}$ in any hereditary class of graphs with sublinear separators.
\end{corollary}

The idea of the algorithm from Theorem~\ref{thm-main} is quite simple: We consider the sets $X_1 , \ldots, X_m$
from the definition of fractional treewidth-$f$-fragility, extend them to suitable supersets $Y_1$, \ldots, $Y_m$, and
argue that for $i=1,\ldots, m$, any admissible set in $G \setminus X_i$ disjoint from $Y_i$ is also admissible in $G$, and that
for some $i$, the weight of the heaviest admissible set in $G \setminus X_i$ disjoint from $Y_i$ is within the factor of $1-1/k$ from the optimal one. 
The construction of the sets $Y_1$, \ldots, $Y_m$ is based on the existence of orientations with bounded outdegrees that represent all short paths, a result of independent interest that we present in Section~\ref{sec-dist}.

Let us remark one can develop the idea of this paper in further directions. 
Dvo\v{r}\'ak proved in~\cite{Dvorak_new}(via a substantially more involved argument) that every monotone maximization problem expressible in first-order logic admits a $\ptas$ in any efficiently fractionally treewidth-fragile class of graphs. 
Note that this class of problems is incomparable with the one considered in this paper (e.g., \textsc{Maximum Induced Forest} is not expressible in the first-order logic, while \textsc{Maximum Independent Set} consisting of vertices belonging to triangles is expressible in the first-order logic but does not fall into the scope of the current paper).

Finally, it is worth mentioning that our results only apply to maximization problems. 
We were able to extend the previous uses of fractional treewidth-fragility by giving a way to handle dependencies over any bounded distance.
However, for the minimization problems, we do not know whether fractional treewidth-fragility is sufficient even for the distance-$1$ problems. 
For a simple example, consider the \textsc{Minimum Vertex Cover} problem in fractionally treewidth-fragile graphs, or more generally in hereditary classes with sublinear separators. 
While the unweighted version can be dealt with by the local search method~\cite{Har-PeledQ17}, we do not know whether there exists a $\ptas$ for the weighted version of this problem.

\section{Paths and orientations in graphs with bounded expansion}\label{sec-dist}
For $r \in \mathbb{Z}^+_0$, a graph $H$ is an \emph{$r$-shallow minor} of a graph $G$ if $H$ can be obtained from a subgraph of $G$ by contracting pairwise vertex-disjoint connected subgraphs, each of radius at most $r$. 
For a function $f \colon \mathbb{Z}^+ \to \mathbb{Z}^+$, a class $\mathcal{G}$ of graphs has \emph{expansion bounded} by $f$ if for all non-negative integers $r$, all $r$-shallow minors of graphs from $\mathcal{G}$ have average degree at most $f(r)$. 
A class has bounded expansion if its expansion is bounded by some function $f$.
The theory of graph classes with bounded expansion has been developed in the last 15 years, and the concept has found many algorithmic and structural applications; see~\cite{NesetrilM12} for an overview. 
Crucially for us, this theory includes a number of tools for dealing with short paths. 
Moreover, as we have pointed out before, all hereditary graph classes with sublinear separators~\cite{DvorakN16} as well as all fractionally treewidth-fragile classes~\cite{Dvorak16} have bounded expansion.

Let $\vec{G}$ be an orientation of a graph $G$, i.e, $uv$ is an edge of $G$ if and only if the directed graph $\vec{G}$ contains
at least one of the directed edges $(u,v)$ and $(v,u)$; note that we allow $\vec{G}$ to contain both of them at the same time, and
thus for the edge $uv$ to be oriented in both directions.  We say that a directed graph $\vec{H}$ with the same vertex
set is a \emph{$1$-step fraternal augmentation of $\vec{G}$} if $\vec{G}\subseteq \vec{H}$, for all distinct edges
$(x,y),(x,z)\in E(\vec{G})$, either $(y,z)$ or $(z,y)$ is an edge of $\vec{H}$, and for each edge $(y,z)\in E(\vec{H})\setminus E(\vec{G})$,
there exists a vertex $x\in V(\vec{G})\setminus \{y,z\}$ such that $(x,y),(x,z)\in E(\vec{G})$.  
That is, to obtain $\vec{H}$ from $\vec{G}$, for each pair of edges $(x,y),(x,z)\in E(\vec{G})$ we add an edge
between $y$ and $z$ in one of the two possible directions (we do not specify the direction, but in practice we would choose
directions of the added edges that minimize the maximum outdegree of the resulting directed graph).
For an integer $a\ge 0$,
we say $\vec{F}$ is an \emph{$a$-step fraternal augmentation of $\vec{G}$} if there exists a sequence
$\vec{G}=\vec{G}_0,\vec{G}_1,\ldots,\vec{G}_a=\vec{F}$ where for $i=1,\ldots, a$, $\vec{G}_i$ is a $1$-step fraternal augmentation of $\vec{G}_{i-1}$.
We say $\vec{F}$ is an $a$-step fraternal augmentation of an undirected graph $G$ if $\vec{F}$ is an $a$-step fraternal augmentation
of some orientation of $G$.  A key property of graph classes with bounded expansion is the existence of fraternal augmentations with bounded
outdegrees.  Let us remark that whenever we speak about an algorithm returning an $a$-step fraternal augmentation $\vec{H}$ or taking one as an input,
this implicitly includes outputing or taking as an input the whole sequence of $1$-step fraternal augmentations ending in $\vec{H}$.
\begin{lemma}[Ne\v{s}et\v{r}il and Ossona de Mendez~\cite{NesetrilM08a}]\label{lemma-frat}
For every class $\GG$ with bounded expansion, there exists a function $d:\mathbb{Z}^+_0\to\mathbb{Z}^+$
such that for each $G\in\GG$ and each non-negative integer $a$, the graph $G$ has an $a$-step fraternal augmentation
of maximum outdegree at most $d(a)$.  Moreover, such an augmentation can be found in time $O(d(a)|V(G)|)$.
\end{lemma}

As shown already in~\cite{NesetrilM08a}, fraternal augmentations can be used to succintly represent distances between
vertices of the graph.  For the purposes of this paper, we need a more explicit representation by an orientation of the
original graph (without the additional augmentation edges).  By a \emph{walk} in a directed graph $\vec{G}$,
we mean a sequence $W=v_0v_1v_2\ldots v_b$ such that for $i=1,\ldots,b$, $(v_{i-1},v_i)\in E(\vec{G})$ or $(v_i,v_{i-1})\in E(\vec{G})$;
that is, the walk does not have to respect the orientation of the edges.
The walk $W$ is \emph{inward directed} if for some $c\in\{0,\ldots,b\}$, we have $(v_i,v_{i+1})\in E(\vec{G})$ for $i=0,\ldots,c-1$
and $(v_i,v_{i-1})\in E(\vec{G})$ for $i=c+1,\ldots,b$.  For a positive integer $r$, an orientation $\vec{G}$
of a graph $G$ \emph{represents $(\le\!r)$-distances} if for each $u,v\in V(G)$ and each $b\in\{0,\ldots,r\}$,
the distance between $u$ and $v$ in $G$ is at most $b$ if and only if $\vec{G}$ contains an inward-directed
walk of length at most $b$ between $u$ and $v$.  Note that given such an orientation with bounded maximum outdegree for a fixed $r$,
we can determine the distance between $u$ and $v$ (up to distance $r$) by enumerating all (constantly many)
walks of length at most $r$ directed away from $u$ and away from $v$ and inspecting their intersections.

Our goal now is to show that graphs from classes with bounded expansion admit orientations with bounded maximum outdegree
that represent $(\le\!r)$-distances.  Let us define a more general notion used in the proof of this claim,
adding to the fraternal augmentations the information about the lengths of the walks in the original graph represented
by the added edges.
A \emph{directed graph with $(\le\!r)$-length sets} is a pair $(\vec{H},\ell)$, where $\vec{H}$ is a directed graph
and $\ell$ is a function assigning a subset of $\{1,\ldots,r\}$ to each \emph{unordered} pair $\{u,v\}$ of vertices of $\vec{H}$,
such that if neither $(u,v)$ nor $(v,u)$ is an edge of $\vec{H}$, then $\ell(\{u,v\})=\emptyset$.
We say that $(\vec{H},\ell)$ is an \emph{orientation} of a graph $G$ if $G$ is the underlying undirected graph of $\vec{H}$
and $\ell(\{u,v\})=\{1\}$ for each $uv\in E(G)$.
We say that $(\vec{H},\ell)$ is an \emph{$(\le\!r)$-augmentation} of $G$ if $V(\vec{H})=V(G)$, for each $uv\in E(G)$
we have $1\in\ell(\{u,v\})$, and for each $u,v\in V(G)$ and $b\in \ell(\{u,v\})$ there exists a walk of length $b$ from $u$ to $v$ in $G$.
Let $(\vec{H}_1,\ell_1)$ be another directed graph with $(\le\!r)$-length sets.  We say $(\vec{H}_1,\ell_1)$ is a \emph{1-step fraternal augmentation}
of $(\vec{H},\ell)$ if $\vec{H}_1$ is a $1$-step fraternal augmentation of $\vec{H}$ and for all distinct $u,v\in V(\vec{H})$ and $b\in\{1,\ldots,r\}$,
we have $b\in\ell_1(\{u,v\})$ if and only if $b\in\ell(\{u,v\})$ or there exist $x\in V(\vec{H})\setminus\{u,v\}$, $b_1\in \ell(\{x,u\})$,
and $b_2\in \ell(\{x,v\})$ such that $(x,u),(x,v)\in E(\vec{H})$ and $b=b_1+b_2$.  Note that a $1$-step fraternal augmentation of
an $(\le\!r)$-augmentation of a graph $G$ is again an $(\le\!r)$-augmentation of $G$.  The notion of an $a$-step fraternal augmentation
of a graph $G$ is then defined in the natural way, by starting with an orientation of $G$ and peforming the $1$-step fraternal augmentation operation
$a$-times.  Let us now restate Lemma~\ref{lemma-frat} in these terms (we just need to maintain the edge length sets, which can be done with
$O(a^2)$ overhead per operation).
\begin{lemma}\label{lemma-lensets}
Let $\GG$ be a class of graphs with bounded expansion, and let $d:\mathbb{Z}^+_0\to\mathbb{Z}^+$ be the function from Lemma~\ref{lemma-frat}.
For each $G\in\GG$ and each non-negative integer $a$, we can in time $O(a^2d(a)|V(G)|)$ construct
a directed graph with $(\le\!a+1)$-length sets $(\vec{H},\ell)$ of maximum outdegree at most $d(a)$
such that $(\vec{H},\ell)$ is an $a$-step fraternal augmentation of $G$.
\end{lemma}

Let $(\vec{H},\ell)$ be an $(\le\!r)$-augmentation $(\vec{H},\ell)$ of a graph $G$.
For $b\le r$, a \emph{length $b$ walk} in $(\vec{H},\ell)$ is a tuple
$(v_0v_1\ldots v_t, b_1,\ldots, b_t)$, where $v_0v_1\ldots v_t$ is a walk in $\vec{H}$,
$b_i\in\ell(\{v_{i-1},v_i\}$ for $i=1,\ldots,t$, and $b=b_1+\ldots+b_t$.  Note that
if there exists a length $b$ walk from $u$ to $v$ in $(\vec{H},\ell)$, then there also
exists a walk of length $b$ from $u$ to $v$ in $G$.
We say that $(\vec{H},\ell)$ \emph{represents $(\le\!r)$-distances} in $G$ if for all vertices $u,v\in V(G)$ at distance $b\le r$ from one another,
$(\vec{H},\ell)$ contains an inward-directed length $b$ walk between $u$ and $v$.
Next, we show that this property always holds for sufficient fraternal augmentations.

\begin{lemma}\label{lemma-augment}
Let $G$ be a graph and $r$ a positive integer and let $(\vec{H},\ell)$ be a directed graph with $(\le\!r)$-length sets.
If $(\vec{H},\ell)$ is obtained as an $(r-1)$-step fraternal augmentation of $G$, then
it represents $(\le\!r)$-distances in $G$.
\end{lemma}
\begin{proof}
For $b\le r$, consider any length $b$ walk $W=(v_0v_1\ldots v_t, b_1,\ldots, b_t)$
in an $(\le\!r)$-augmentation $(\vec{H}_1,\ell_1)$ of $G$, and let $(\vec{H}_2,\ell_2)$ be a $1$-step augmentation
of $(\vec{H}_1,\ell_1)$.  Note that $W$ is also a length $b$ walk between $v_0$ and $v_t$ in $(\vec{H}_2,\ell_2)$.
Suppose that $W$ is not inward-directed in $(\vec{H}_1,\ell_1)$, and thus there exists $i\in\{1,\ldots,t-1\}$ such that
$(v_i,v_{i-1}),(v_i,v_{i+1})\in E(\vec{H}_1)$.  By the definition of $1$-step fraternal augmentation,
this implies $b_i+b_{i+1}\in \ell_2(v_{i-1},v_{i+1})$, and thus
$(v_0\ldots v_{i-1}v_{i+1}\ldots v_t, b_1,\ldots,b_i+b_{i+1},\ldots b_t)$ is a length $b$ walk from $v_0$ to $v_t$
in $(\vec{H}_2,\ell_2)$.

Let $(\vec{G}_0,\ell_0)$, \ldots, $(\vec{G}_{r-1},\ell_{r-1})$ be a sequence of $(\le\!r)$-augmentations of $G$,
where $(\vec{G},\ell_0)$ is an orientation of $G$, $(\vec{G}_{r-1},\ell_{r-1})=(\vec{H},\ell)$, and
for $i=1,\ldots, r-1$, $(\vec{G}_i,\ell_i)$ is a $1$-step fraternal augmentation of $(\vec{G}_{i-1},\ell_{i-1})$.
Let $u$ and $v$ be any vertices at distance $b\le r$ in $G$, and let $P$ be a shortest path between
them.  Then $P$ naturally corresponds to a length $b$ walk $P_0$ in $(\vec{G}_0,\ell_0)$.
For $i=1,\ldots, r-1$, if $P_{i-1}$ is inward-directed, then let $P_i=P_{i-1}$, otherwise
let $P_i$ be a length $b$ walk in $(\vec{G}_i,\ell_i)$ obtained from $P_{i-1}$ as described in the previous paragaph.  Since each application
of the operation decreases the number of vertices of the walk, we conclude that $P_{r-1}$ is
an inward-directed length $b$ walk between $u$ and $v$ in $(\vec{H},\ell)$.
Hence, $(\vec{H},\ell)$ represents $(\le\!r)$-distances in $G$.
\end{proof}

Next, let us propagate this property back through the fraternal augmentations by orienting some of the edges
in both directions.  We say that $(\vec{H},\ell)$ is an \emph{$a$-step fraternal superaugmentation} of a graph $G$
if there exists an $a$-step fraternal augmentation $(\vec{F},\ell)$ of $G$ such that $V(\vec{F})=V(\vec{H})$,
$E(\vec{F})\subseteq E(\vec{H})$ and for each $(u,v)\in E(\vec{H})\setminus E(\vec{F})$, we have $(v,u)\in E(\vec{F})$.
We say that $(\vec{F},\ell)$ is a \emph{support} of $(\vec{H},\ell)$.

\begin{lemma}\label{lemma-back}
Let $G$ be a graph and $r$ a positive integer and let $(\vec{H},\ell)$ be an $(\le\!r)$-augmentation of $G$ of maximum outdegree $\Delta$
representing $(\le\!r)$-distances.  For $a\ge 1$, suppose that $(\vec{H},\ell)$ is an $a$-step fraternal superaugmentation of $G$.
Then we can in time $O(r^2\Delta|V(G)|)$ obtain an $(a-1)$-step fraternal superaugmentation of $G$ representing $(\le\!r)$-distances,
of maximum outdegree at most $(r+1)\Delta$.
\end{lemma}
\begin{proof}
Let $(\vec{F},\ell)$ be an $a$-step fraternal augmentation of $G$ forming a support of $(\vec{H},\ell)$,
obtained as a $1$-step fraternal augmentation of an $(a-1)$-step fraternal augmentation $(\vec{F}_1,\ell_1)$ of $G$.
Let $(\vec{H}_1,\ell_1)$ be the $(a-1)$-step fraternal superaugmentation of $G$
obtained from $(\vec{F}_1,\ell_1)$ as follows:
\begin{itemize}
\item For all distinct vertices $y,z\in V(G)$ such that $(y,z),(z,y)\in E(\vec{H})$, $(y,z)\in E(\vec{F}_1)$,
and $(z,y)\not\in E(\vec{F}_1)$, we add the edge $(z,y)$.
\item For each edge $(y,z)\in E(\vec{H})$ and integer $b\in\ell(\{y,z\})\setminus\ell_1(\{y,z\})$,
we choose a vertex $x\in V(G)\setminus\{y,z\}$ such that $(x,y),(x,z)\in E(\vec{F}_1)$ and
$b=b_1+b_2$ for some $b_1\in \ell_1(\{x,y\})$ and $b_2\in\ell_1(\{x,z\})$, and add the edge
$(y,x)$.  Note that such a vertex $x$ and integers $b_1$ and $b_2$ exist, since $b$ was added to $\ell(\{y,z\})$
when $(\vec{F},\ell)$ was obtained from $(\vec{F}_1,\ell_1)$ as a $1$-step fraternal augmentation.
\end{itemize}
Each edge $(y,x)\in E(\vec{H}_1)\setminus E(\vec{H})$ arises from an edge $(y,z)\in E(\vec{H})$
leaving $y$ and an element $b\in \ell(\{y,z\})\setminus\ell_1(\{y,z\})$, and each such pair contributes at most
one edge leaving $y$.  Hence, the maximum outdegree of $\vec{H}_1$ is at most $(r+1)\Delta$.

Consider a length $b$ inwards-directed walk $(v_0v_1\ldots v_t,b_1,\ldots,b_t)$ in $\vec{H}$, for any $b\le r$.
Then $\vec{H}$ contains a length $b$ inwards-directed walk from $v_0$ to $v_t$ obtained by natural edge replacements:
For any edge $(y,z)\in E(\vec{H})$ of this walk and $b'\in \ell_i(\{y,z\})$, the construction described above
ensures that if $(y,z)\not\in E(\vec{H}_1)$ or $b'\not\in \ell_1(\{y,z\})$, then
there exists $x\in V(G)\setminus\{y,z\}$ such that $(y,x),(x,z)\in E(\vec{H}_1)$ and $b'=b''+b'''$ for some
$b''\in \ell_1(\{x,y\})$ and $b'''\in\ell_1(\{x,z\})$, and we can replace the edge $(y,z)$ in the walk
by the edges $(y,x)$ and $(x,z)$ of $E(\vec{H}_1)$.
Since $\vec{H}$ represents $(\le\!r)$-distances in $G$, this transformation shows that so does $\vec{H}_1$.
\end{proof}
We are now ready to prove the main result of this section.
\begin{lemma}\label{lem:orient}
For any class $\GG$ with bounded expansion, there exists a function $d':\mathbb{Z}^+\to\mathbb{Z}^+$
such that for each $G\in\GG$ and each positive integer $r$, the graph $G$ has an orientation with maximum outdegree at
most $d'(r)$ that represents $(\le\!r)$-distances in $G$. Moreover, such an orientation can be found in time $O(r^2d'(r)|V(G)|)$.
\end{lemma}
\begin{proof}
Let $d$ be the function from Lemma~\ref{lemma-frat}, and let $d'(r)=(r+1)^{r-1}d(r-1)$.
By Lemma~\ref{lemma-lensets}, we obtain an $(r-1)$-step fraternal augmentation $(\vec{H},\ell)$ of $G$ of maximum outdegree at most $d(r-1)$.
By Lemma~\ref{lemma-augment}, $(\vec{H},\ell)$ represents $(\le\!r)$-distances in $G$.  Repeatedly applying Lemma~\ref{lemma-back},
we obtain a $0$-step fraternal superaugmentation $(\vec{G},\ell_0)$ of $G$ of maximum outdegree at most $d'(r)$ representing $(\le\!r)$-distances.
Clearly, $\vec{G}$ is an orientation of $G$ of maximum outdegree at most $d'(r)$ representing $(\le\!r)$-distances.
\end{proof}

\section{Approximation schemes}

Let us now prove Theorem~\ref{thm-main}.  To this end, let us start with a lemma to be applied to
the sets arising from fractional treewidth-fragility.

\begin{lemma}\label{lemma-avoid}
Let $\vec{G}$ be an orientation of a graph $G$ with maximum outdegree $\Delta$.  Let $A$ be a set of vertices of $G$
and for a positive integer $c$, let $\{R_v:v\in A\}$ be a system of subsets of $A$ such that each vertex belongs to at most $c$ of the
subsets.  For $X\subseteq V(G)$ and a positive integer $r$, let $D_{\vec{G},r}(X)$ be the union of the sets $R_v$
for all vertices $v\in V(G)$ such that $\vec{G}$ contains a walk from $v$ to $X$ of length at most $r$ directed away from $v$.
For a positive integer $k$, let $X_1$, \ldots, $X_m$ be a system of subsets of $V(G)$ such that each vertex belongs
to at most $\frac{m}{c(\Delta+1)^rk}$ of the subsets.  For any assignment $w$ of non-negative weights to vertices of $G$,
there exists $i\in\{1,\ldots,m\}$ such that $w(A\setminus D_{\vec{G},r}(X_i))\ge (1-1/k)w(A)$.
\end{lemma}
\begin{proof}
For a vertex $z\in A$, let $B(z)$ be the set of vertices reachable in $\vec{G}$ from vertices $v\in A$ such that $z\in R_v$
by walks of length at most $r$ directed away from $v$.  Note that $|B(z)|\le c(\Delta+1)^r$ and that for each $X\subseteq V(G)$,
we have $z\in D_{\vec{G},r}(X)$ if and only if $B(z)\cap X\neq \emptyset$.

Suppose for a contradiction that for each $i$ we have $w(A\setminus D_{\vec{G},r}(X_i))<(1-1/k)w(A)$, and thus
$w(D_{\vec{G},r}(X_i))>w(A)/k$.  Then
\begin{align*}
\frac{m}{k}w(A)&<\sum_{i=1}^m w(D_{\vec{G},r}(X_i))=\sum_{i=1}^m \sum_{z\in D_{\vec{G},r}(X_i)} w(z)=\sum_{i=1}^m \sum_{z\in A:B(z)\cap X_i\neq\emptyset} w(z)\\
&\le \sum_{i=1}^m\sum_{z\in A} w(z)|B(z)\cap X_i|=\sum_{z\in A} w(z) \sum_{i=1}^m |B(z)\cap X_i|\\
&=\sum_{z\in A} w(z)\sum_{x\in B(z)} |\{i\in\{1,\ldots,m\}:x\in X_i\}|\le \sum_{z\in A} w(z)\sum_{x\in B(z)} \frac{m}{c(\Delta+1)^rk}\\
&=\sum_{z\in A} w(z)|B(z)| \frac{m}{c(\Delta+1)^rk}\le \sum_{z\in A} w(z)\frac{m}{k}=\frac{m}{k}w(A),
\end{align*}
which is a contradiction.
\end{proof}

Next, let us derive a lemma on admissibility for $(\le\!r)$-distance determined problems.

\begin{lemma}\label{lemma-admis}
For a positive integer $r$, let $\vec{G}$ be an orientation of a graph $G$ representing $(\le\!r)$-distances.
For a set $X\subseteq V(G)$, let $Y_{\vec{G},r}(X)$ be the set of vertices $y$ such that
$\vec{G}$ contains a walk from $y$ to $X$ of length at most $r$ directed away from $y$.  For any $(\le\!r)$-distance determined problem,
a set $B\subseteq V(G)\setminus Y_{\vec{G},r}(X)$ is admissible in $G$ if and only if it is admissible in $G-X$.
\end{lemma}
\begin{proof}
Since the problem is $(\le\!r)$-distance determined, it suffices to show that
$\min(r,d_G(u,v))=\min(r,d_{G-X}(u,v))$ holds for all $u,v\in B$.  Clearly,
$d_G(u,v)\le d_{G-X}(u,v)$, and thus it suffices to show that if the distance between $u$ and $v$ is $G$ is $b\le r$,
then $G-X$ contains a walk of length $b$ between $u$ and $v$.  Since $\vec{G}$ represents $(\le\!r)$-distances,
there exists an inward-directed walk $P$ of length $b$ between $u$ and $v$ in $\vec{G}$.  Since $u,v\not\in Y_{\vec{G},r}(X)$,
we have $V(P)\cap X=\emptyset$, and thus $P$ is also a walk of length $b$ between $u$ and $v$ in $G-X$.
\end{proof}

We are now ready to prove the main result.

\begin{proof}[Proof of Theorem~\ref{thm-main}]
Let $d'$ be the function from Lemma~\ref{lem:orient} for the class $\GG$.
Let us define $h(r,c)=c(d'(r)+1)^r$.  The algorithm is as follows.
Since $\GG$ is $q$-efficiently fractionally treewidth-$f$-fragile, in time $q(|V(G)|)$
we can find sets $X_1, \ldots, X_m\subseteq V(G)$ such that each vertex belongs to at most $\frac{m}{h(r,c)k}$
of them, and for each $i$, a tree decomposition of $G-X_i$ of width at most $f(h(r,c)k,|V(G)|)$.
Clearly, $m\le q(|V(G)|)$.  Next, using Lemma~\ref{lem:orient}, we find an orientation $\vec{G}$ of $G$ that
represents $(\le\!r)$-distances.  Let $Y_{\vec{G},r}$ be defined as in the statement of Lemma~\ref{lemma-admis}.
Since the problem is $(g,p)$-tw-tractable problem, for each $i$ we can in time $p(|V(G)|)\cdot g(f(h(r,c)k,|V(G)|))$
find a subset $A_i$ of $V(G)\setminus Y_{\vec{G},r}(X_i)$ admissible in $G-X_i$ of largest weight.
By Lemma~\ref{lemma-admis}, each of these sets is admissible in $G$; the algorithm return the heaviest
of the sets $A_1$, \ldots, $A_m$.

As the returned set is admissible in $G$, it suffices to argue about its weight.  Let $A$ be a heaviest admissible set
in $G$.  Let $\{R_v\subseteq A:v\in A\}$ be the system of subsets from the definition of $c$-near-monotonicity,
and let $D_{\vec{G},r}$ be defined as in the statement of Lemma~\ref{lemma-avoid}.  By the definition of $c$-near-monotonicity,
for each $i$ the set $A\setminus D_{\vec{G},r}(X_i)$ is admissible in $G$.  Since $v\in R_v$ for each $v\in A$, we have
$Y_{\vec{G},r}(X_i)\subseteq D_{\vec{G},r}(X_i)$, and thus by Lemma~\ref{lemma-admis}, $A\setminus D_{\vec{G},r}(X_i)$
is also admissible in $G-X_i$, and by the choice of $A_i$, we have $w(A_i)\ge w(A\setminus D_{\vec{G},r}(X_i))$.
By Lemma~\ref{lemma-avoid}, we conclude that for at least one $i$, we have $w(A_i)\ge (1-1/k)w(A)$, as required.
\end{proof}

\end{document}